\documentclass[letterpaper,11pt]{article}
\usepackage[english]{babel}
\usepackage[ansinew]{inputenc}
\usepackage{fontenc}
\usepackage{amsfonts}
\usepackage{amsmath}
\usepackage{amsthm}
\usepackage{enumerate}
\usepackage{textcomp}
\usepackage{geometry}
\usepackage{pst-tree}
\usepackage{eurosym}
\usepackage{mathrsfs}
\usepackage[algo2e]{algorithm2e}
\usepackage{caption}
\usepackage[colorlinks=true,linkcolor=red,citecolor=blue]{hyperref}
\newcommand{\n}{\noindent}
\newtheorem{remark}{Remark}[section]
\newtheorem{proposition}{Proposition}[section]
\newtheorem{property}{Property}
\newtheorem{conclusion}{Conclusion}
\newtheorem{example}{Example}[section]

\begin{document}

\date{}
\title{A recursive algorithm for a pipeline maintenance scheduling
problem}
\author{Assia Boumahdaf\thanks{%
e-mail: assia.boumahdaf@gmail.com} \hspace{0.5cm} Michel Broniatowksi%
\thanks{
e-mail: michel.broniatowski@upmc.fr} \\
Laboratoire de Statistique Th\'eorique et Appliqu\'ee, \\ Universit\'e Pierre et Marie Curie}
\maketitle

\begin{abstract}
This paper deals with the problem of preventive maintenance (PM) scheduling of pipelines subject to external corrosion defects. The preventive maintenance strategy involves an inspection step at some epoch, together with a repair schedule. This paper proposes to determine the repair schedule as well as an inspection time minimizing the maintenance cost. This problem is formulated as a binary integer non-linear programming model and we approach it under a decision support framework. We derive a polynomial-time algorithm that computes the optimum PM schedule and suggests different PM strategies in order to assist practitioners in making decision.
\end{abstract}

\section{Introduction}
\label{section:intro}

Gas pipelines are facilities intended for the transport of natural gas at high pressure. Pipelines carry natural gas from the extraction area to export area and are buried under the ground in inhabited zones. A major threat for their integrity is  metal-loss corrosion. To maintain safe pipeline condition, preventive maintenance (PM) is performed. Poor pipeline management  can cause leaks and leads to human and environmental damage, as well as monetary loss. 

As states Zhou in \cite{Zhou2010} a pipeline management program consists firstly in detecting the corrosion defects with appropriate equipment, secondly in evaluating the probability of failure based on the primary inspection results, and lastly in repairing the defects if necessary. Defects which do not call into question system integrity are not immediately repaired, and will be considered for PM, which consists of identifying the next inspection time, together with a repair schedule. This paper deals with the PM problem in gas pipeline from an economic point of view. We aim to investigate the optimal PM schedule which minimizes the operational cost. Large scale maintenance activities have significant cost. There are not only the costs related to inspection and repairs, but also the cost due to  production losses during  maintenance. When maintenance activities are conducted, gas flow in pipelines must be interrupted for security measures, generating a significant out-of-service cost. Hence, the operational cost estimate includes the cost due to  inspection, repairs, and the cost due to out-of-service of gas pipelines.

Several models and methodologies establishing optimal PM schedules can be found in the literature. A number of papers have been published using a reliability approach. In the corroding gas pipeline context, we refer the reader to \cite{SunMaMorris2009} for a recent survey about this subject. In their paper, they address the problem of predicting the reliability of pipelines with imperfect repairs in order to assist pipeline operators in making the most appropriate PM decision. Hong \cite{Hong1999} uses also a reliability analysis to estimate the probability of failure, together with optimal inspection and maintenance schedules, based on minimization of the total cost including  inspection,  repairs, and the cost of failure. In a more general context, Kallen \cite{Kallen2005} determines the optimal inspection and replacement policy which minimize the expected average costs using an adaptive Bayesian decision model.

To the best of our knowledge, there have been no previous economic and deterministic studies on PM of gas pipelines. However, this issue arises in a wide variety of areas, one of which is power plants. This is an important issue because a failure in a power station may cause an overall breakdown, and significant customer dissatisfaction. Canto \cite{Canto2011} considers the problem of selecting the period for which the facilities should be taken offline for preventive maintenance. He models this problem with a binary integer linear program, and solved it using optimization software. The same author \cite{Canto2008} solves the problem using Benders decomposition technique. Megow \cite{Megow2011} derives a model for the problem of turnaround scheduling. They propose an optimization algorithm within a decision support framework that computes the schedule minimizing the resource cost.

Preventive maintenance problems also arise in the medical field. Vlah Jeri\'c and Figueira \cite{VlahJeric2012} consider the issue of scheduling medical treatments for patients. They formulate the problem as a binary integer programming model, and solve it using a meta-heuristic algorithm. \cite{Schimmelpfeng2012} propose a decision support system for resource scheduling. Chern \textit{et al.} \cite{Chern2008} consider health examination scheduling. They model this problem using a binary integer programming model, and propose solutions based on a heuristic algorithm.

There are many other fields that deal with PM scheduling problems; in the military context, we can mention \cite{Joo2009}, who have developed a dynamic approach for scheduling PM of engine modules. Maintenance
scheduling problems involving machines have been investigated by Hadidi \textit{et al.} \cite{Hadidi2012} and Keha \textit{et al.} \cite{Keha2009} for instance, and in a paper factory by Garg \textit{et al.}\cite{Garg2013}.

In this paper, we assume that we have collected information about defects which were not handled during the primary maintenance management program described above by Zhou \cite{Zhou2010}. This information consists of an acceptable limit date for repairs; we call this date \textit{the deadline}. In the rest of this paper, the repairs not handled during the primary inspection, with their associated deadlines, will be called the primary repair schedule, and are the starting point of our study. This paper models the economic preventive maintenance scheduling problem in gas pipelines as a binary integer non-linear programming model, and presents an algorithmic solution based on dynamic programming. This algorithm is performed in polynomial time and computes the global solution; it proposes also a class of alternative solutions which may assist industrial personnel in making decisions.

The remainder of this paper is organized as follows. The model is formulated in Section~\ref{section1:problem description}. Section~\ref{section2:mathematical model} models the PM problem by a binary nonlinear program. Section~\ref{section3:algorithm} focuses on the algorithmic solution. Computational results are presented in Section~\ref{section4:simulations}, and we conclude in Section~\ref{section:conclusion}

\section{Problem description}
\label{section1:problem description}

After a primary inspection defined at time $t=0$, a long term horizon is fixed, which we denote $T^{\ast } \in  \mathbb{N}^{*}$; no repairs will be handled after $T^{\ast}$, which may be thought of as the maximal time before the next inspection. The next inspection may be happen at time $t$, $ t \in \{  1,2, \ldots, T^{\ast }\}$. 
Define the inspection interval by $\Delta t=(0,t]$. During the primary inspection, a number of corrosion defects are detected. Some of them, considered as unacceptable for the safety of the pipeline, are repaired immediately. Those which do not call into question the pipeline integrity are not immediately repaired; to each of these is associated a deadline corresponding to the limit date for repair. So, the pipeline manager must plan repair activities no later than this deadline for each specific defect. Beyond this date, the safety of the pipeline is seriously compromised. Assume that we have knowledge of these deadlines and the number of defects to be repaired for each of these dates. We may thus define what we have called in Section~\ref{section:intro} the \textit{primary repair schedule}:

\begin{equation}  
\label{primary PM}
\mathscr{P}=\left\lbrace(n_1, T_1), \ldots ,(n_N,T_N) \right\rbrace,
\end{equation}

\n 
where $N>0$ is the number of different deadlines. For $i=1,\ldots N $, $n_{i}$ is the number of defects to be repaired before their deadline $T_{i}$, with $n_{i}>0$, $T_{i} < T^{\ast }$ and $T_{1}<\ldots <T_{N}$ (with $ T_{0}=0 $). From the primary repair schedule $\mathscr{P}$ defined by \eqref{primary PM}, we seek  the next optimal inspection time, denoted by $t^{\ast }$, belonging to the set $\{1,\ldots ,T^{\ast }\}$, and the optimal repair program within the
inspection interval $\Delta t^{\ast }$ that minimizes the operational cost.

\subsection{Operational cost}

Preventive maintenance activities include inspection and repairs. Moreover, gas pipelines must be interrupted for safety during  maintenance activities. These activities generate a significant out-of-service cost corresponding to the financial loss due to the inactivity of the gas pipeline. The costs to be considered are the inspection cost, denoted by  $C_{insp}$, the repair cost $C_{rep}$, and the out-of-service cost, $C_{out}$. Hence, the operational cost is

\begin{equation}  
\label{eq1}
C_{tot} = C_{insp} + C_{rep} + C_{out}.
\end{equation}

\begin{remark}
The out-of-service cost $C_{out}$ represents the financial loss due to repair activities. The inspection cost also takes  into account an out-of-service cost. 
\end{remark}

Let $\mathscr{P}=\left\{ (n_{i},T_{i}),i=1,\ldots ,N\right\} $ be the primary repair schedule and let $t$, $t \in \{ 1, \ldots , T^* \}$ be any fixed inspection time. The total cost \eqref{eq1} depends on both $t$ and the repair plan within the inspection interval $\Delta t$. Denote by $N_{t}$ the number of deadlines $T_{i}$, $i\in \{1,\ldots ,N\}$ within $\Delta t$, \textit{i.e.},

\begin{equation}
\label{N_t}
N_{t}=card\left\{ i\in \{1,\ldots ,N\}:T_{i}\leq t\right\} .
\end{equation}

\n 
When the next inspection is planned at time $t$, the operational cost related to $\mathscr{P}$, $C(t, \mathscr{P})$ is given by

\begin{equation}  
\label{total cost C*}
C(t, \mathscr{P})= C_{insp}(t) + \sum_{i=1}^{N_{t}}
n_iC_{rep}(T_i) + \sum_{i=1}^{N_{t}} C_{out}(T_i).
\end{equation}

Let us first make some assumptions about the costs defined above. We will consider two economic parameters: the discount rate and the inflation rate. For a given initial cost, for example, with initial inspection cost denoted $C_{inp}^{0}$, the inspection cost at time $t>0$ will be given by

\begin{equation*}
C_{insp}(t) = C_{inp}^{0}\times \left( \dfrac{1+r_i}{1+r_d}\right)^t,
\end{equation*}

\n 
where $r_{i}$ and $r_{d}$ are respectively the inflation rate and  discount rate. Moreover, we assume that $r_{i}<r_{d}$. Thus, the function $t\mapsto C_{insp}(t)$ is decreasing. This will also be true for $C_{rep}(.)$ and $C_{out}(.)$.

The PM scheduling problem is twofold; the first part consists of selecting the next inspection time within the set $\{ 1, 2, \ldots , T^* \}$, and the second, to plan the repair schedule within $\Delta t$ in order to minimize the total cost. In the following, the optimal solution will be denoted by $(t^*, \mathscr{P}_{t}^{*})$. This problem is highly combinatorial. It consists of finding the next  inspection time within $\{1, \ldots , T^*\}$, and the least expensive repair program among all possible programs achieved from $\mathscr{P}$,  under the constraint that a defect cannot be repaired after its deadline. Thus, finding the exact solution in a short time cannot be expected. However as we will seen in the forthcoming section, by exploiting properties of the model (Section~\ref{sub-model properties}), we will be able to reduce the space of feasible solutions. This combinatorial optimization problem will be modeled using a binary nonlinear programming model in Section~\ref{sub-minlp formulation}, and an effective algorithmic solution will be proposed in Section~\ref{sub-Algo}.

\section{Mathematical model}
\label{section2:mathematical model}

\subsection{Some properties of the model}
\label{sub-model properties}

Before formulating the mathematical model as a binary integer nonlinear program, it is worth noting a number of simple properties. We will use the fact that we can only do repairs early, and not late; the repair and the out-of-service
costs decrease with time; anticipating repairs where other repairs are planned does not add an out-of-service cost. These properties will allow us to reduce the space of repair schedules to explore to $(1-2^{N+1})/(1-2) -1$. Let us introduce some notation that we will use in the following. Denotes by $\mathscr{D}_{k}$ the set of deadlines (including $0$) up to $T_{k}$, for $k=1,\ldots ,N-1$,

\begin{equation*}  \label{deadline}
\mathscr{D}_{k} = \left\lbrace T_1, \ldots , T_k \right\rbrace \cup \{0 \}.
\end{equation*}

\n 
We denote by $\mathscr{D}$ (the subscript $N$ will be omitted) the set of all deadlines (including $0$), \textit{i.e.}, 

\begin{equation*}
\mathscr{D}=\{T_{1},\ldots ,T_{N}\}\cup \{0\}.
\end{equation*}

\begin{property}
All defects with same deadline will be repaired at the same time.
\end{property}

\n 
For example, $3$ defects have to repaired before year $15$. Thus, we will repair the 3 defects at the same time, either at year $0$ (during the primary inspection), at year $1$, $2$, and so on, up to year $15$. This means that we will not split up repairs, since this action increases the total cost, adding a repair cost and/or an out-of-service cost. This fact is expressed in the following proposition.

\begin{proposition}
\label{prop1} 
Let $t$ be an inspection time, $t \in \{ 1, \ldots  T^* \}$, and let $\mathscr{P}=\{(n_{i},T_{i}),i=1,\ldots ,N\}$ be the primary repair program with total cost $ C(t, \mathscr{P})$. Let $N_{t}$ be the number of deadlines within $\Delta t$ defined by \eqref{N_t}. For any set of defects $n_{1},n_{2},\ldots ,n_{N}$, the new repair program defined after splitting a set is more expensive than $\mathscr{P}$.
\end{proposition}

\begin{proof}
We will prove this proposition by considering only one set $n_k$, with $k \in \{ 1, \ldots  N\}$. The result will be applied when considering several sets. We consider the $k$th set of defects $n_k$, such that $n_k \geq 2$, which must be repaired before $T_k$. Suppose that we have split the set $n_k$ into two other sets, $n_k^{'}$ and $n_k^{''}$, such that $n_k = n_k^{'} + n_k^{''}$, where $n_k^{'} \geq 1$ is the number of repairs that will be performed early at  time $\tau$ with $\tau < T_k$. Without loss of generality, assume that the $n_k^{''} \geq 1$ defects will be repaired at time $T_k$. We assume initially that $k \leq N_{t}$, i.e., the $n_k$ defects are within $\Delta t$. 

$\bullet$ If $\tau \notin \mathscr{D}_{k-1}$, define the new repair schedule as
\begin{equation}
\label{new PM1 prop1}
  \mathscr{P}_{1}=\left\lbrace (n_1,T_1),\ldots ,(n_k^{'}, \tau),\ldots,(n_k^{''}, T_k), (n_{k+1}, T_{k+1}), \ldots ,(n_N, T_N)\right\rbrace.
\end{equation}

\n
Noting that planning repairs before deadlines add an out-of-service cost,  the total cost of $\mathscr{P}_{1}$ is given by

\begin{equation*}
\label{Ctot PM1 prop1}
  C(t, \mathscr{P}_{1}) = C_{insp}(t) + \sum_{i=1}^{N_{t}}n_iC_{rep}(T_i) - n_k^{'}C_{rep}(T_k)+ n_k^{'}C_{rep}(\tau) + \sum_{i=1}^{N_{t}} C_{out}(T_i) + C_{out}(\tau).
\end{equation*}

\n
Using the fact that $C_{rep}(T_k) < C_{rep}(\tau)$, we get 

\begin{equation*}
C(t, \mathscr{P}) - C(t, \mathscr{P}_{1}) = n_k^{'}[C_{rep}(T_k) - C_{rep}(\tau)] - C_{out}(\tau) < 0.
\end{equation*}

$\bullet$ If $\tau \in \mathscr{D}_{k-1}$, then there exists $i$ such that $\tau = T_i< T_k$. The new repair plan is given by 

\begin{equation}
\label{new PM2 prop1}
  \mathscr{P}_2 =\left\lbrace (n_1,T_1),\ldots, (n_i+ n_k^{'},T_i),\ldots ,(n_k^{''}, T_k), (n_{k+1}, T_{k+1}), \ldots ,(n_N, T_N)\right\rbrace.
\end{equation}

\n
Remark next that  repairs made earlier than absolutely necessary at a time where other repairs are planned does not add an out-of-service cost. Thus, the total cost of $\mathscr{P}_{2}$ is given by
\begin{equation*}
\label{Ctot PM2 prop1}
C(t, \mathscr{P}_{2}) = C_{insp}(t) + \sum_{i=1}^{N_{t}}n_iC_{rep}(T_i) - n_k^{'}C_{rep}(T_k)+ n_k^{'}C_{rep}(T_i) + \sum_{i=1}^{N_{t}}C_{out}(T_i).
\end{equation*}

\n
Thus,
\begin{equation*}
C(t ,\mathscr{P}) - C(t, \mathscr{P}_{2}) = n_k^{'}[C_{rep}(T_k) - C_{rep}(T_i)] < 0.
\end{equation*}

\n
Assume now that $k > N_{t}$, i.e., the set $n_k$ is not within $\Delta t$, but we have the opportunity to perform early the set $n_k^{'}$ within $\Delta t$. In this case, all repair schedules that we can define will be more expensive than $\mathscr{P}$ because we add to $\mathscr{P}$ an out-of service cost depending on whether $\tau \in \mathscr{D}_{N_{t}}$ or not, and on the monotonicity of $C_{rep}(.)$. Therefore, in all cases the cost of $\mathscr{P}$ is less expensive than all repair plans defined by splitting.
\end{proof}

\begin{remark}
We have supposed that the $n_{k}^{^{\prime \prime }}$ repairs take place at $T_{k}$. We could have decided to  repairs early, but the associated repair schedule would be more expensive than \eqref{new PM1 prop1} and \eqref{new PM2 prop1}.
\end{remark}

\begin{remark}
We have added an out-of-service cost at time $\tau $ to \eqref{Ctot PM1 prop1} because no repair was planned at time $\tau $. Thus, moving forward repairs from deadlines generates a cost due to the unavailability of gas from the pipeline.
\end{remark}

\begin{remark}
\label{rem-restricted PM} 
For a given inspection time $t$, when the set of  $n_k$ defects is not in  $\Delta t$ (i.e., $k > N_{t}$), moving some  repairs into the inspection interval is more expensive than repairing the set $n_k$ at time $T_k$, because we add an out-of-service cost, and $C_{rep}(.)$ is decreasing.
\end{remark}

\begin{remark}
We  also observe that the first part of the proof suggests that  repairs from after deadlines done early generate more expensive repair plan schedules than $\mathscr{P}$. This property will be proved below.
\end{remark}

\begin{conclusion}
\label{cl1}
Proposition~\ref{prop1} allows for all split scenarios to be rejected, thus  reducing the space of feasible solutions. Moreover, we can restrict our attention to deadlines within $\Delta t$ in order to build the optimal repair senario.
\end{conclusion}

\begin{property}
There is no monetary advantage to repair defects outside the times $\mathscr{D}$.
\end{property}

\n 
For example, we fix the horizon time $T^* = 20$, and the next inspection at $t=18$. Denote by $\mathscr{P} = \left\lbrace (1, 3), (2,10), (3, 15)\right\rbrace$ the primary
repair schedule with total cost $C(t,\mathscr{P})$ given by

\begin{equation*}
C(18,\mathscr{P}) = C_{insp}(18)+ C_{rep}(3)+2 C_{rep}(10)+3 C_{rep}(15)+ C_{out}(3)+
C_{out}(10)+C_{out}(15).
\end{equation*}

\n 
Let $\mathscr{P}_1$ be a new repair schedule where we have repaired early the third set of defects with deadline $15$ at year $\tau$ such that $\tau \notin \mathscr{D}_2 = \left\lbrace 3, 10 \right\rbrace \cup \{0\}$ and $\tau < 15$. The total cost of $\mathscr{P}_1$ is

\begin{equation*}
C(18, \mathscr{P}_1) = C_{insp}(18)+ C_{rep}(3)+2C_{rep}(10)+3
C_{rep}(\tau)+ C_{out}(3)+ C_{out}(10)+C_{out}(\tau).
\end{equation*}

\n Hence,
\begin{equation*}
C(18, \mathscr{P}) - C(18, \mathscr{P}_1) = 3\underbrace{\left( C_{rep}(15) -
C_{rep}(\tau) \right)}_{< 0} + \underbrace{\left(C_{out}(15) -C_{out}(\tau)
\right)}_{< 0}.
\end{equation*}

\n This leads to the following Proposition.

\begin{proposition}
\label{prop2} 
Let $t$ be a fixed inspection time, with $t \in \{ 1, \ldots , T^* \}$, and let $\mathscr{P}$ be the primary repair program with total cost $C(t, \mathscr{P})$. Then, the total cost related to any
repair schedules such that some defects are repaired before their deadline, is more expensive than $\mathscr{P}$.
\end{proposition}

\n 
\begin{proof}
We prove the result for only one set of defects repaired early. According to Proposition~\ref{prop1}, repairs with same deadline are handled at the same time. Let $\mathscr{P}_{1}$ be a repair schedule such that $n_{k}$ repairs are done early at time $\tau$ such that $\tau < T_k$.  Assume furthermore that $\tau \notin \mathscr{D}_{k-1}$. The total cost of $\mathscr{P}_{1}$ is

\begin{equation*}
 C(t, \mathscr{P}_{1}) = C_{insp}(t)+\sum_{i=1,i\neq k}^{N_{t}} n_i C_{rep}(T_i)+n_kC_{rep}(\tau) +\sum_{i=1, i\neq k}^{N_{t}}C_{out}(T_i)+ C_{out}(\tau).
\end{equation*}

\n
Thus,
\begin{equation*}
  C(t, \mathscr{P}) -  C(t, \mathscr{P}_1) = n_k[C_{rep}(T_k)-C_{rep}(\tau)] - C_{out}(\tau)<0.
\end{equation*}
This result can then be applied when considering several blocks, in order to conclude the proof.
\end{proof}  

\begin{conclusion}
\label{cl2}
Proposition~\ref{prop2} allows for all repair schedules that include repairs
done before their deadline to be rejected.
\end{conclusion}

\begin{remark}
At this stage, we have moved from a very large space of schedules to a space of $(N_{t}+1)!$ possible schedules  for a given inspection time $t$.  Indeed, using Conclusion~\ref{cl1} and Conclusion~\ref{cl2}, we can repair the set $n_{1}$ either at time $T_{1}$ or earlier at time $T_{0}$ (two choices); the set $n_{2}$ can be repaired either at time $T_{2}$, or either at time $T_{1}$ or either at time $T_{0}$. For the last set $n_{N_t}$, we have $N_{t}+1$ possibilities. Thus, we obtain $(N_{t}+1)!$ plans to choose between if the next inspection is planned at time $t$.
\end{remark}

The next proposition states that the space of repair schedules may be reduced to $2^{N_{t}}$ for a given inspection time $t$. Denote by $\mathscr{P}_{t}^{j}$ the repair program such that the $N_{t}$ repairs (the last set of defects within $\Delta t$) has been done early at time $T_{j}$, $j=0,\ldots ,N_{t}-1$,
\begin{equation}  
\label{P*_j}
\mathscr{P}_{t}^{j} = \left\lbrace (n_0, T_0), (n_1, T_1), \ldots,
(n_j+n_{N_{t}}, T_j),(n_{j+1}, T_{j+1}), \ldots,(n_{N_{t} -1}, T_{N_{t} -1})
\right\rbrace, \quad n_0=0.
\end{equation}

\n 
For $j=0,\ldots ,N_{t}-1$, the associated total cost is given by
\begin{equation*}
C_{t}^{j} = C_{insp}(t)+n_{1}C_{rep}(T_{1})+\ldots
+(n_{j}+n_{N_{t}})C_{rep}(T_{j})+\ldots +n_{N_{t}-1}C_{rep}(T_{N_{t}-1})+\sum_{i=1}^{N_{t}-1}C_{out}(T_{i}).
\end{equation*}

\begin{proposition}
\label{prop3} 
The number of feasible solutions is $2^{N_{t}}$ for a
given inspection time $t \in \{ 1, \ldots ,T^* \}$.
\end{proposition}

\begin{proof}
We shall prove this proposition by induction on $N_{t}$. For $N_{t} = 1$, i.e., the inspection is planned so that there is only one set $n_1$ (with deadline $T_1$) within $\Delta t$. In this case, the proposition is proved. We shall prove the result for $N_{t} = 2$ and suppose that the $n_1$ repairs are planned. We shall prove that there are two repair scenarios for the set of $n_2$ defects.  We consider two cases. In the first, the $n_1$ repairs take place at time $T_1$; in the second, they are moved to time $T_0=0$.\\

\n
\textit{First case.} We have $\mathscr{P}_{t} = \{ (n_1, T_1), (n_2, T_2)\}$ with total cost $C(t, \mathscr{P}_{t})$ given by
\begin{equation*}
C(t, \mathscr{P}_{t}) = C_{insp}(t)+ \sum_{i=1}^{2}n_i C_{rep}(T_i)+ \sum_{i=1}^{2} C_{out}(T_i).
\end{equation*}

\n
The total cost associated with $\mathscr{P}_{t}^{1} = \{ (n_1+n_2, T_1)\}$ is 
\begin{equation*}
C_{t}^{1} = C_{insp}(t)+ (n_1+n_2)C_{rep}(T_1)+  C_{out}(T_1),
\end{equation*}

\n
and the cost related to $\mathscr{P}_{t}^{0} = \{ (n_2, T_0),(n_1, T_1)\}$ is 
\begin{equation*}
C_{t}^{0} = C_{insp}(t)+ n_2 C_{rep}(T_0) + n_1 C_{rep}(T_1)+ C_{out}(T_1).
\end{equation*}

\n
When comparing $C_{t}^{1}$ with $C_{t}^{0}$, we obtain
\begin{equation*}
C_{t}^{1} - C_{t}^{0} = n_2\left[ C_{rep}(T_1)-C_{rep}(T_0) \right] < 0,
\end{equation*}

\n
and thus can rule out the repair plan $\mathscr{P}_{t}^{0}$. When comparing $C(t, \mathscr{P}_{t})$ with $C_{t}^{1}$, we have
\begin{equation*}
C(t, \mathscr{P}_{t}) - C_{t}^{1} = n_2\underbrace{\left[ C_{rep}(T_2)-C_{rep}(T_1) \right]}_{ < 0} + C_{out}(T_2).
\end{equation*}

\n
In this case, we cannot conclude which is the best program in terms of minimal cost. Consequently, for a given inspection time $t$, we may repair the set $n_2$ either at its deadline $T_2$ ($\mathscr{P}_t$) or early at time $T_1$ ($\mathscr{P}_t^1$). \\

\n 
\textit{Second case.} Set $\tilde{\mathscr{P}}_{t} =\{(n_1, T_0), (n_2, T_2) \}$ with total cost $C(t, \tilde{\mathscr{P}}_{t})$ given by
\begin{equation*}
C(t, \tilde{\mathscr{P}}_{t}) = C_{insp}(t)+ n_1 C_{rep}(T_0) + n_2 C_{rep}(T_2)+ C_{out}(T_2).
\end{equation*}

\n
The cost related to $\tilde{\mathscr{P}}_{t}^{1} = \{ (n_1, T_0),(n_2,T_1)\}$ is 
\begin{equation*}
C_{t}^{1} = C_{insp}(t)+ n_1 C_{rep}(T_0) + n_2 C_{rep}(T_1)+ C_{out}(T_1),
\end{equation*}

\n
and the total cost associated with $\tilde{\mathscr{P}}_{t}^{0}= \{ (n_1+n_2, T_0)\}$ is 
\begin{equation*}
C_{t}^{0} = C_{insp}(t)+ (n_1+n_2) C_{rep}(T_0).
\end{equation*}

\n
Since $C(t, \tilde{\mathscr{P}}_{t}) - C_{t}^{1} = n_2\left[ C_{rep}(T_2) - C_{rep}(T_1)\right] + C_{out}(T_2) - C_{out}(T_1) < 0$, we can rule out the schedule $\tilde{\mathscr{P}}_{t}^{1}$. Furthermore, when comparing $C(t, \tilde{\mathscr{P}}_{t})$ with $C_{t}^{0}$, we have
\begin{equation*}
C(t, \tilde{\mathscr{P}}_{t}) - C_{t}^{0} = n_2\left[ C_{rep}(T_2)-C_{rep}(T_0) \right] + C_{out}(T_2),
\end{equation*}
which does not allow us to conclude anything, so the proposition is proved for $N_{t}=2$.

Now, let $t$ be an inspection time such that $N_{t} = k+1$. Assuming that the set $n_1, \ldots , n_k$ are planned, we will consider only one configuration. We suppose that the sets of $n_1, \ldots, n_k$ defects take place respectively at times $T_1, \ldots , T_k$. We shall prove that for a given inspection time $t$, there are two choices for positioning the $(k+1)$th repairs. Using the fact that $ C_{rep}(T_k)< C_{rep}(T_{k-1}) < \ldots < C_{rep}^{0}$, and the fact that repairing early when other repairs are planned does not add an out of service cost, yields
 \begin{equation*}
  C_{t}^{k} < C_{t}^{k-1}< \ldots < C_{t}^{0}.
\end{equation*}

\n
Thus, we can rule all repair plans such that the $n_{k+1}$ repairs are done at an  earlier date than $T_k$. There remain only two repair schedules to compare, $\mathscr{P}_{t}$ and $\mathscr{P}_{t}^{k}$,
 with respective total costs $C(t,\mathscr{P}_{t})$ and $C_{t}^{k}$.
 \begin{equation*}
  C(t,\mathscr{P}_{t}) - C_{t}^{k} = n_k\underbrace{\left[ C_{rep}(T_{k+1}) - C_{rep}(T_k)\right]}_{< 0}  + C_{out}(T_{k+1}).
\end{equation*}
We cannot conclude which is the best. Consequently, we may either plan the $n_{k+1}$ repairs at year $T_k$ (schedule $\mathscr{P}_{t}^{k}$) or   at year $T_{k+1}$ (schedule $\mathscr{P}_{t}$).  
\end{proof}

\begin{conclusion}
\label{cl3}
Proposition~\ref{prop3} states that for a given inspection time $t$, the number of schedules to consider is $2^{N_{t}}$. For a given set of defects $n_k$, $k \in \{ 1, \ldots , N_t \}$, there are two repair strategies. Either the set $n_{k}$ is handled at its deadline $T_k$, or earlier at the previous deadline, where repairs are already planned.
\end{conclusion}

The next proposition plays a crucial role in designing the algorithm which will solve the optimization problem \eqref{Opt Problem} in the forthcoming section. Let $\tilde{\mathscr{P}}$ be any repair schedule achieved from $\mathscr{P}$. 

\begin{proposition}
\label{prop4} 
For any repair schedule $\tilde{\mathscr{P}}$, the function $s \in (0,T^*) \mapsto C_{tot}(s,\tilde{\mathscr{P}})$ is not convex and has local minima at points $ s =\{ T_{j}-1, \; j=1, \ldots , N  \} \cup \{ T^* \}$. The optimal inspection $t^*$ is found at one of these.
\end{proposition}

\begin{proof}
The total cost function $C_{tot}(., \tilde{\mathscr{P}})$ is the sum  of a decreasing function $C_{insp}(.)$ and a step function $(C_{rep} + C_{out})(.)$ with jump discontinuities at $T_j -1$ for $ j=1, \ldots , N$. Then, the total cost function $C_{tot}(., \mathscr{P})$ is increasing within $[T_j -1, T_j]$ and decreasing within $[T_{j}, T_{j+1} -1]$ and $[T_{N_t}, T^*]$. Thus $s \mapsto C_{tot}(s, \mathscr{P})$ cannot be convex and has local minima  at $T_j -1$, $j=1, \ldots , N$. Therefore, the global minimum is found at one of these dates. 
\end{proof} 

\begin{remark}
Note that if $T_{1} = 1$, $T_1 -1$ should not be considered as a candidate for the next inspection because it would  coincide with the primary inspection.
\end{remark}

\begin{conclusion}
\label{cl4}
Proposition~\ref{prop4} suggests that the optimal solution for the PM schedule problem should be chosen for an inspection time among $\{ T_j-1, j=1, \ldots , N \} \cup \{  T^* \}$ instead of $\{1, \ldots , T^*\}$. This implies that the space of feasible solutions may be reduced to $\sum_{k=1}^{N} 2^{k}$ repair plans to consider.
\end{conclusion}

\n 
This proposition pushes us to solve the PM scheduling problem in a support decision framework, by building an algorithm that proposes a set of repair schedules related to inspection times within $t \in \{T_{j}-1,j=1,\ldots ,N\}\cup \{ T^*\}$, which includes the optimal PM schedule $(t^{*},\mathscr{P}_{t}^{*})$, in order to support decisions of pipeline managers.

\subsection{Formulating PM as an integer programming problem}
\label{sub-minlp formulation}

The problem described above can be formulated as a binary integer nonlinear programming model, taking into account all the previous propositions, and certain constraints that we shall now
describe. \vspace{0.3cm}

We are looking for an optimal inspection time $t^{\ast}$ such that $t^{\ast}\in \{T_1 - 1 ,\ldots ,T_{N}-1, T^{\ast}\}$, and a repair schedule within $\Delta t^{\ast}$ that minimizes the total cost. In the following, we introduce two decision variables $a_i$ and $b_j$, for $i=1, \ldots ,N+1$, and $j=0, \ldots ,
N$, defined by:
\begin{equation*}
a_i = \left\lbrace 
\begin{array}{lll}
1 \;\; \mbox{if an inspection is planned in year}\;\; T_i -1 &  &  \\ 
0 \;\; \mbox{otherwise,} &  & 
\end{array}
\right.
\end{equation*}

\begin{equation*}
a_{N+1} = \left\lbrace 
\begin{array}{lll}
1 \;\; \mbox{if an inspection is planned in year}\;\; T^* &  &  \\ 
0 \;\; \mbox{otherwise,} &  & 
\end{array}
\right.
\end{equation*}

\n 
and
\begin{equation*}
b_j = \left\lbrace 
\begin{array}{lll}
1 \;\; \mbox{if repairs are planned in year}\;\; T_j &  &  \\ 
0 \;\; \mbox{otherwise,} &  & 
\end{array}
\right.
\end{equation*}

\n 
where $T_0=0$. Denote by $a$ and $b$ the variables 
\begin{equation*}
a = (a_1, \ldots, a_{N+1}) \quad \mbox{and} \quad b = (b_0, \ldots, b_N).
\end{equation*}

\n
Both vectors must satisfy certain constraints, which we now describe. The first ensures that on the time interval $(0,T^{\ast }]$, there is only one inspection after the primary inspection. Note that if $T_1 = 1$, we cannot plan the next inspection at time $0$, that is during the primary inspection; thus, define the variable $\alpha$ such that
\begin{equation*}
\alpha = \left\lbrace 
\begin{array}{lll}
1 \;\; \mbox{if} \;\; T_1 \neq 1   \\ 
0 \;\; \mbox{otherwise.}  
\end{array}
\right.
\end{equation*}

\n
Therefore, the constraint stating that there is only one inspection on $(0,T^*]$ is reflected as: 
\begin{equation*}
\alpha a_1 + \sum_{j=2}^{N+1} a_j = 1.
\end{equation*}

\n 
For example, if $T^{\ast }=10$, $N=4$ and $T_1 \neq 1$, the variable $a=(0,0,0,0,1)$ of length $5$
means that an inspection is planned at time $T^{\ast}$. However, if $T_1 = 1$, the vector $a$ has length $4$, and $a=(0,1,0,0)$ means that the inspection is planned at year $T_3 - 1$. \\

\n 
A second constraint encodes the fact that we cannot plan repairs simultaneously at times $T_0$ and $T_1$. Thus,
\begin{equation*}
b_0 + b_1 = 1.
\end{equation*}

\n 
Indeed, Proposition~\ref{prop3} states that there are two options when planning repairs. Suppose that we decide to repair the $n_{1}$ defects at time $T_{1}$ (at their deadline). Then, the $n_{2}$ defects may be repaired either at time $T_{2}$ (their deadline) or time $T_{1}$, but not at time $T_{0}=0$. Suppose now that the $n_{1}$ repairs were moved at time $T_{0}$, then the ${2}$  repairs may be planned at time $T_{2}$ or at time $T_{0}$, but not at time $T_{1}$.
Since defects with the same deadline have to be repaired at the same time, we have at most $N$ sets of repairs:
\begin{equation*}
1 \leq \sum_{j=0}^{N} b_j \leq N.
\end{equation*}

\n
The PM scheduling problem is then the following:
\begin{equation}  
\label{Opt Problem}
\begin{array}{ll}
\mbox{Minimize} \quad \;\; C_{tot}(a_1,\ldots, a_{N+1}, b_0,\ldots , b_N)%
\vspace{0.3 cm} &  \\ 
\mbox{subject to} \left\{ 
\begin{array}{llll}
a_i \in \{0,1\} \; \mbox{for all}\; i= 1,\ldots, N+1 \vspace{0.2 cm} &  &  & 
\\ 
b_j \in \{0,1\} \; \mbox{for all}\; j= 0,\ldots, N \vspace{0.2 cm} &  &  & 
\\ 
\alpha a_1 + \sum_{i=2}^{N+1} a_i = 1, \;\; \alpha \in \{ 0,1 \} \vspace{0.2 cm}  \\ 
b_0 + b_1 = 1 \vspace{0.2 cm} &  &  &  \\ 
1 \leq \sum_{j= 0}^{N} b_j \leq N, &  &  & 
\end{array}
\right. & 
\end{array}%
\end{equation}

\n 
where the objective function is given by
\begin{align}  
\label{objective function}
\begin{split}
C(&a_1,\ldots, a_{N+1}, b_0,\ldots , b_N) =  \alpha a_1C_{insp}(T_1-1) + \sum_{i=2}^{N}a_iC_{insp}(T_i -1) + a_{N+1}C_{insp}(T^*) \\
& + \sum_{j=0}^{N-1}b_j(\Pi_{i=1}^{j}(1-a_i))(n_j + (1-b_{j+1})(n_{j+1} +(1-b_{j+2})( n_{j+2} + \ldots +(n_N(1-b_N))\ldots)) C_{rep}(T_j) \\
&\qquad + b_N \times (\Pi_{i=1}^{N}(1-a_i))\times n_NC_{rep}(T_N)  + \sum_{j=1}^N b_j (\Pi_{i=1}^{j}(1-a_i))C_{out}(T_j),
\end{split}%
\end{align}

\n 
with the conventions $n_0 = 0$ and $\Pi_{i=1}^{T_0}(1-a_i) = 1$. We illustrate this objective function with the following example.

\begin{example}
We fix the time horizon $T^{\ast }$ as $T^{\ast }=10$ and the number of deadlines $N$ as $N=3$. Set $T_{1}=2$, $T_{2}=4$ and $T_{3}=8$, and for all  $i \in \{{1,2,3}\}$, we set $n_{i}=1$. Thus, the primary PM \eqref{primary PM} is given by 
\begin{equation*}
\mathscr{P}=\{(1,T_1),(1,T_2),(1,T_3)\}.
\end{equation*}

\n 
Suppose that we want an inspection to take place at year $t=T_3 - 1 = 7$. Since $T_1 \neq 1$ then $\alpha =1$, and thus
\begin{equation*}
a=(0,0,1,0).
\end{equation*}

\n
Suppose furthermore that we want to do repairs with deadline $T_{2}=4$ early, at year $T_{1}=2$; then, the new PM schedule $\mathscr{P}_{t}^{1}$ \eqref{P*_j}, with $N_{t} = 2$, is $\mathscr{P}_{t}^{1}=\{(1+1,T_{1})\}$; thus the vector $b$ is defined by 
\begin{equation*}
b=(0,1,0,0).
\end{equation*}

\n 
To perform the total cost of this program, $a$ and $b$ are substituted into \eqref{objective function}, with 
\begin{equation*}
a=(0,0,1,0) \qquad \mbox{and}\qquad b=(0,1,0,1).
\end{equation*}

\n 
The calculation of the inspection cost for $\mathscr{P}_{t}^{1}$ gives
\begin{align*}
C_{insp}(a) &= a_1C_{insp}(T_1-1) + a_2C_{insp}(T_2-1)+a_3C_{insp}(T_3-1)+a_4C_{insp}(T^*)\\
            &= a_3C_{insp}(T_3-1) = 1 \times C_{insp}(7). 
\end{align*}

\n
The calculation of the reparation cost is somewhat complicated because it has to take into account whether or not the deadlines are within the inspection interval $\Delta t = (0,T_3-1]$, and the total number of repairs for each deadline. Noting that $T_0$, $T_1$ and $T_2$ are within $\Delta t$, we  multiply the variables $b_0$, $b_1$ and $b_2$ respectively by $\Pi_{i=1}^{0}(1-a_i) = 1$, $(1-a_1)=1$ and $(1-a_1)(1-a_2)=1$. Since $T_3 \notin \Delta t$, the variable $b_3$ will be multiplied by $(1-a_1)(1-a_2)(1-a_3) = 0$, and does not contribute to the total cost, even if $b_3 = 1$. The total number of repairs for each deadline is expressed by a non-linear term that involves $n_j$ and $(1-b_j)$. Thus, the calculation of the repair cost yields
\begin{align*}
C_{rep}(a,b) &= b_1\left((1-a_1)\right)(n_1 + (1-b_{2})(n_{2} + (n_3(1-b_3)))C_{rep}(T_1) \\
             &= 1 \times \left((1-0)\right) (n_1 + (1-0)(n_2 + n_3(1-1)))C_{rep}(T_1) \\
             &= (n_1+n_2)C_{rep}(T_1).
\end{align*}

\n
The out-of-service cost is given by
\begin{align*}
C_{out}(a,b) &= b_1 \times (1-a_1) C_{out}(T_1)=1 \times C_{out}(T_1).
\end{align*}

\n 
We obtain
\begin{equation*}
C(a, b) = C_{insp}(6) + 2C_{rep}(T_1) + C_{out}(T_1),
\end{equation*}

\n 
which is the total cost of $\mathscr{P}_{t}^{1}$.
\end{example}

\begin{remark}
Proposition~\ref{prop4} states that the optimal solution of \eqref{Opt Problem} is achieved for an inspection time belonging to $\{ T_1 -1, \ldots , T_N -1, T^* \}$. Then, the size of the space of feasible solutions may be reduced to $\sum_{k=0}^{N} 2^k - 1$. Indeed, if $t = T_1 - 1$, then $N_t = 0$ and there is no repair plan to explore. If $t = T_2-1$, then $N_t=1$ and there are $2^1$ repair plans to consider. If $t=T_3-1$, $2^2$ repair schedules have to be considered, which yields the result.
\end{remark}

\section{Solution-finding strategy}
\label{section3:algorithm}

The PM scheduling problem consists of finding an optimal inspection time $t^{\ast}$ such that $t^{\ast}\in\{T_1-1,\ldots , T_{N}-1, T^*\}$, and a repair schedule within $\Delta t^{\ast}$ that minimizes the total cost. 
The search for the optimal solution requires considering $(1-2^{N+1})/(1-2) -1$ possible repair plans. We will see in this section how to reduce this number to $(N+1)(N+2)/2-1$. We outline an algorithm in the context of the support decision framework, which computes the optimal PM schedule, together with alternative solutions, in order to propose to the pipeline manager a set of ''good'' solutions. A naive way to solve the problem is to look through, for a fixed inspection time $t \in \{ T_j-1, j=1, \ldots ,N , T^* \}$, all of the $2^{N_{t}}$ feasible solutions, and save any of those with the best objective function. For large values of $N$, such a method is not efficient. To quantify the efficiency of our algorithm, we have developed two other algorithms related to the problem \eqref{Opt Problem}). The first investigates all repair schedules when $t \in \{ T_j-1, j=1, \ldots ,N , T^* \}$, and the second, which will serve as a benchmark, looks through all repair plans when $t \in \{ 1, 2, \ldots, T^* \}$. This will be developed in Section~\ref{section4:simulations}.

\subsection{Construction of the algorithm}
\label{sub sec algo}

The aim of this section is to develop a strategy which builds a search tree and returns the optimal schedule for a given inspection time. In order to design our algorithm, we shall use the following example. Over a time horizon $T^{\ast }=20$, we consider $N=3$ sets of defects with cardinality $n_{1},n_{2},n_{3}$, respectively, with associated deadlines $T_{1}$, $T_{2}$ and $T_{3}$, and we will fix the inspection time as $t = T^*$. Note that this search tree will not  necessarily find the global optimum. In order to find that, we have to build all  $N+1$ search trees, each  related to inspection times, and compare the associated total costs.

\subsubsection{First and second generations}

The inspection time is $t = T^*$, thus the number of set of defects that will be considered is $N_{t} = 3$. We need to look through $2^3=8$ repair schedules. We shall represent the various repair plans by a tree, where each branch represents one strategy. For example, the branch $br(2)$ corresponds to  the repair plan $\{(n_1,T_1),(n_2+n_3,T_2) \}$ (see Figure~\ref{figure1}). 

The strategy is to look through all $8$ repair plans in order to discard the most expensive branches.
After removing some branches, a final tree will remain, corresponding to potential solutions; the least expensive will be the optimal for fixed inspection times (but not necessarily the global optimum). Note that since $N_{t} = 3$, the final tree will have three generations. This methodology will allow us to draw rules for an iterative construction of the final search tree. Figure~\ref{figure8} depicts such a final tree.

\begin{figure}[h!]
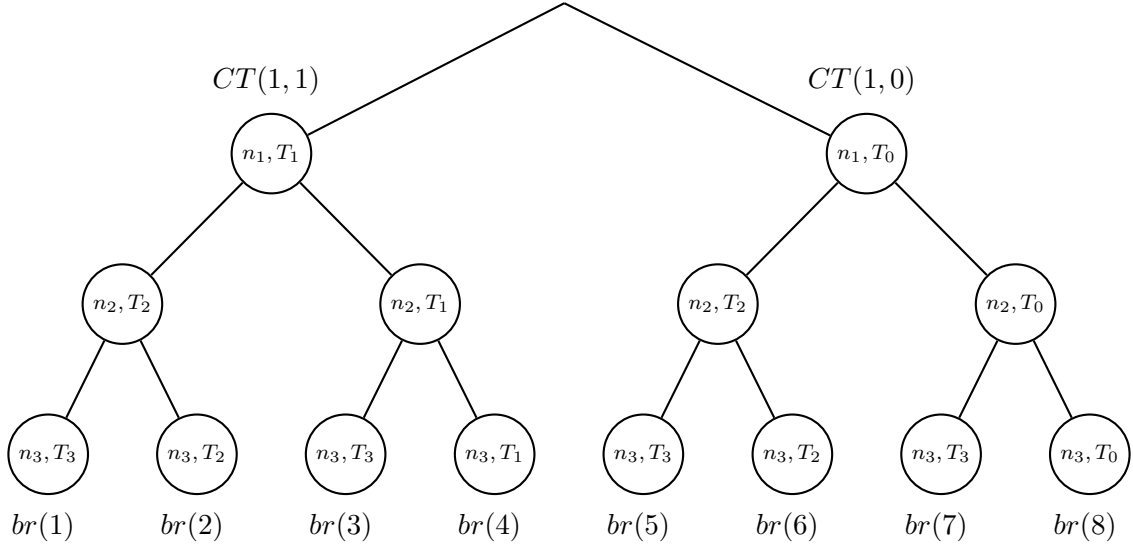

\pstree{\Tr{}}
{ \pstree{\Tcircle{\scriptsize $n_1,T_1$}~[tnpos=a]{$CT(1,1)$}}{\pstree{\Tcircle{\scriptsize $n_2,T_2$}} {\Tcircle{\scriptsize $n_3,T_3$}~{$br(1)$} \Tcircle{\scriptsize $n_3,T_2$}~{$br(2)$}} \pstree{\Tcircle{\scriptsize $n_2,T_1$}}{\Tcircle{\scriptsize $n_3,T_3$}~{$br(3)$} \Tcircle{\scriptsize $n_3,T_1$}~{$br(4)$}}}

\pstree{\Tcircle{\scriptsize $n_1,T_0$}~[tnpos=a]{$CT(1,0)$}}{\pstree{\Tcircle{\scriptsize $n_2,T_2$}} {\Tcircle{\scriptsize $n_3,T_3$}~{$br(5)$} \Tcircle{\scriptsize $n_3,T_2$}~{$br(6)$}} \pstree{\Tcircle{\scriptsize $n_2,T_0$}}{\Tcircle{\scriptsize $n_3,T_3$}~{$br(7)$} \Tcircle{\scriptsize $n_3,T_0$}~{$br(8)$}}}
}
\caption{The search tree representing the $2^{N_{t}}$ PM schedules when $t = T^*$.}
\label{figure1}
\end{figure}

We are first interested in building and comparing the branches $br(1)$, $br(2)$, $br(5)$, $br(6)$. The nodes $(n_{1},T_{0})$ and $(n_{1},T_{1})$ are labeled with the total costs, denoted by $CT(1,0)$ and $CT(1,1)$ respectively, corresponding to the cost related to the \textit{partial plans} $\{(n_{1},T_{0})\}$ and $\{(n_{1},T_{1})\}$(see Figure~\ref{figure1}). The costs $CT(1,0)$ and $CT(1,1)$ are given by
\begin{equation}  \label{CT10}
CT(1,0) = C(t,\{(n_1,T_0)\}) = C_{insp}(t) + n_1C_{rep}(T_0),
\end{equation}

\n and 
\begin{equation}  
\label{CT11}
CT(1,1)=C(t, \{(n_1,T_1)\})= C_{insp}(t)+ n_1C_{rep}(T_1)+ C_{out}(T_1).
\end{equation}

\n The total costs of the branches $br(1)$, $br(2)$, $br(5)$, $br(6)$ are 
\begin{equation*}
\begin{array}{llll}
C(t,\,br(1))=CT(1,1) + n_2C_{rep}(T_2)+ C_{out}(T_2) + n_3C_{rep}(T_3)+
C_{out}(T_3), &  &  &  \\ 
C(t,\,br(5))=CT(1,0) + n_2C_{rep}(T_2)+ C_{out}(T_2) +
n_3C_{rep}(T_3)+C_{out}(T_3), &  &  &  \\ 
C(t,\,br(2))=CT(1,1) + n_2C_{rep}(T_2)+ C_{out}(T_2) + n_3C_{rep}(T_3), & 
&  &  \\ 
C(t,\,br(6))=CT(1,0) + n_2C_{rep}(T_2)+ C_{out}(T_2) + n_3C_{rep}(T_2). & 
&  &  \\ 
&  &  & 
\end{array}%
\end{equation*}

\n When comparing the repair schedules $br(1)$ with $br(5)$, and $br(2)$ with $br(6)$, we get
\begin{align*}
C(t,\,br(1)) - C(t,\,br(5)) &= CT(1,1) - CT(1,0) =  n_1 ( \underbrace{C_{rep}(T_1) - C_{rep}(T_0)}_{ < 0} ) + \underbrace{C_{out}(T_1)}_{> 0} ,
\end{align*}

\n and
\begin{align*}
C(t,\,br(2)) - C(t,\,br(6)) &= CT(1,1) - CT(1,0) =  n_1(C_{rep}(T_1) - C_{rep}(T_0)) + C_{out}(T_1).
\end{align*}

\n At this stage, we cannot conclude as to which is the most expensive program, so we must add a condition on
\begin{equation*}
CT(1,1) - CT(1,0).
\end{equation*}

$\bullet $ If $CT(1,1)<CT(1,0)$, then the node (or the partial repair
program) $\{(n_{1},T_{1})\}$ has the smallest cost function, and we get
\begin{equation*}
C(t,\,br(1)) < C(t,\,br(5)) \qquad \mbox{and} \qquad C(t,\,br(2)) < C(t,\,br(6)).
\end{equation*}

\n  The repair programs $br(5)$ and $br(6)$ are respectively more expensive than $br(1)$ and $br(2);$  the branches $br(5)$ and $br(6)$ are rejected, i.e., the child node $(n_{2},T_{2})$ of $(n_{1},T_{0})$ is discarded. 

$\bullet $ If $CT(1,0)<CT(1,1)$, then the node $(n_{1},T_{0})$ has the smallest cost function
\begin{equation*}
C(t,\,br(1)) > C(t,\,br(5)) \qquad \mbox{and} \qquad C(t,\,br(2)) > C(t,\,br(6)).
\end{equation*}

\n The branches $br(1)$ and $br(2)$ are respectively more expensive than $br(5)$ and $br(6);$ hence  $br(1)$ and $br(2)$ are rejected, i.e., the descendant of $(n_{1},T_{1})$ is discarded, \textit{i.e.}, $(n_{2},T_{2})$.

We have so far compared the branches $br(1)$, $br(2)$, $br(5)$, $br(6)$. It remains to compare $br(3)$ with $br(7)$ and $br(4)$ and $br(8)$. As before, we calculate the total cost of the repair programs $br(3)$, $br(7)$, $br(4)$
and $br(8)$:
\begin{equation*}
\begin{array}{lllll}
C(t,\,br(3)) = CT(1,1) + n_2C_{rep}(T_1) + n_3C_{rep}(T_3) + C_{out}(T_3),%
\vspace{0.1 cm} &  &  &  &  \\ 
C(t,\,br(7)) = CT(1,0) + n_2C_{rep}(T_0) + n_3C_{rep}(T_3) + C_{out}(T_3), 
\vspace{0.1 cm} &  &  &  &  \\ 
C(t,\,br(4)) = CT(1,1) + n_2C_{rep}(T_1) + n_3C_{rep}(T_1),\vspace{0.1 cm}
&  &  &  &  \\ 
C(t,\,br(8)) = CT(1,0) + n_2C_{rep}(T_0) + n_3C_{rep}(T_0). &  &  &  &  \\ 
&  &  &  & 
\end{array}
\end{equation*}

\n When comparing $br(3)$ with $br(7)$, and $br(4)$ with $br(8)$, we obtain
\begin{equation*}
\begin{array}{lll}
C(t^,\,br(3)) - C(t,\,br(7)) = CT(1,1) - CT(1,0) + n_2(\underbrace{
C_{rep}(T_1) - C_{rep}(T_0)}_{ < 0}).\vspace{0.1 cm} &  &  \\ 
C(t,\,br(4)) - C(t,\,br(8)) = CT(1,1) - CT(1,0) + (n_2+
n_3)(C_{rep}(T_1) - C_{rep}(T_0)). &  & 
\end{array}%
\end{equation*}

$\bullet$ If $CT(1,1) < CT(1,0)$, then
\begin{equation*}
C(t,\,br(3)) < C(t,\,br(7)) \qquad \mbox{and} \qquad C(t,\,br(4)) < C(t,\,br(8)).
\end{equation*}
We may discard $br(7)$ and $br(8)$, that is, the descendant of $(n_{1},T_{1})$, i.e., $(n_{2},T_{0})$. \vspace{0.2cm}

$\bullet$ If $CT(1,1) > CT(1,0)$, then
\begin{equation*}
\begin{array}{ll}
C(t,\,br(3)) - C(t,\,br(7)) = (\underbrace{CT(1,1) - CT(1,0)}_{ > 0 }) +
n_2(\underbrace{ C_{rep}(T_1) - C_{rep}(T_0)}_{ < 0}), \vspace{0.1 cm} &  \\ 
C(t,\,br(4)) - C(t,\,br(8)) = (CT(1,1) - CT(1,0)) + (n_2+
n_3)(C_{rep}(T_1) - C_{rep}(T_0)). & 
\end{array}%
\end{equation*}

\n Once again, we cannot conclude as to the most costly program, hence we keep the branches $br(3)$, $br(4)$, $br(7)$ and $br(8)$, and thus the nodes $(n_{2},T_{1})$ and $(n_{2},T_{0})$. 

By combining these observations, we can define a rule in order to build the second generation. We  begin with the search tree depicted in Figure~\ref{figure1}, representing the $2^{N_{t}} = 8$ PM plans. The first generation is composed of the nodes  $(n_{1},T_{1})$ and $(n_{1},T_{0})$, and  labeled respectively with the cost functions $CT(1,1)$ and $CT(1,0)$, defined by \eqref{CT11} and \eqref{CT10}. 

$\bullet $ When $\{(n_{1},T_{1})\}$ has a smaller cost function than $\{(n_{1},T_{0})\}$, \textit{i.e.}, when  $CT(1,1)<CT(1,0)$, we  discard the most expensive branches $br(5)$, $br(6)$, $br(7)$ and $br(8)$, and thus the node $(n_{1},T_{0})$. Only four branches remain, among which the optimal program (for the fixed inspection time), as shown in Figure~\ref{figure2}.

\begin{figure}[h!]
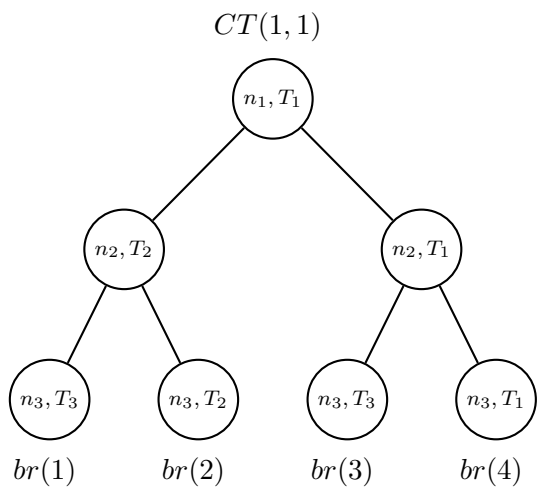

\begin{center}
\pstree{\Tcircle{\scriptsize $n_1,T_1$}~[tnpos=a]{$CT(1,1)$}}  {\pstree{\Tcircle{\scriptsize $n_2,T_2$}} { {\Tcircle{\scriptsize
$n_3,T_3$}~{$br(1)$}  \Tcircle{\scriptsize $n_3,T_2$}~{$br(2)$}}}  \pstree{\Tcircle{\scriptsize $n_2,T_1$}} {{\Tcircle{\scriptsize
$n_3,T_3$}~{$br(3)$} \Tcircle{\scriptsize $n_3,T_1$}~{$br(4)$}}}  }
\caption{Search tree when $CT(1,1) < CT(1,0)$.}
\label{figure2}
\end{center}
\end{figure}

$\bullet $ Under the condition $CT(1,1)>CT(1,0)$ (i.e., the node $(n_{1},T_{0})$ has a smaller cost function than $(n_{1},T_{1})$), the most expensive branches are $br(1)$ and $br(2)$. In this case, only six branches
remain ($br(3),\ldots br(8)$) instead of eight, among which is found the optimal program (see Figure~\ref{figure3}).

\begin{figure}[h!]
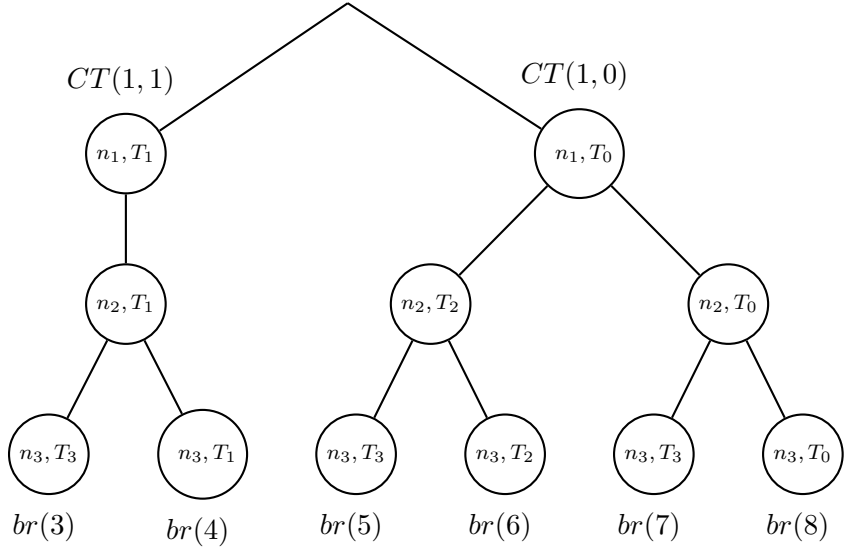

\begin{center}
\pstree{\Tr{}}
{ \pstree{\Tcircle{\scriptsize $n_1,T_1$}~[tnpos=a]{$CT(1,1)$}}
    { \pstree{\Tcircle{\scriptsize $n_2,T_1$}}{\Tcircle{\scriptsize $n_3,T_3$}~{$br(3)$} \Tcircle{    \scriptsize $n_3,T_1$}~{$br(4)$}}}

   \pstree{\Tcircle{ \scriptsize $n_1,T_0$}~[tnpos=a]{$CT(1,0)$}}{\pstree{\Tcircle{\scriptsize $n_2,T_2$}} {\Tcircle{\scriptsize $n_3,T_3$}~{$br(5)$} \Tcircle{\scriptsize $n_3,T_2$}~{$br(6)$}} \pstree{\Tcircle{\scriptsize $n_2,T_0$}}{\Tcircle{\scriptsize $n_3,T_3$}~{$br(7)$} \Tcircle{\scriptsize $n_3,T_0$}~{$br(8)$}}}
}
\caption{Search tree when $CT(1,1) > CT(1,0)$.}
\label{figure3}
\end{center} 
\end{figure}

Now, we are able to define the strategy for designing the second generation. Suppose
that we have built only the first generation, $\{(n_{1},T_{1}),(n_{1},T_{0})\}$, as depicted in Figure~\ref{figure4}. Nodes are labeled with cost function $CT(1,1)$ and $CT(1,0)$ respectively.

\begin{figure}[h!]
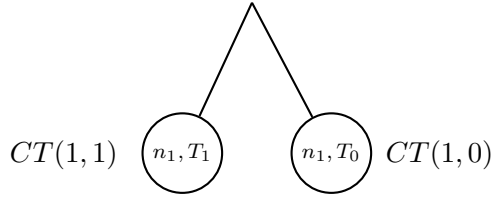

\begin{center}
\pstree{\Tr{}} { \Tcircle{\scriptsize $n_1,T_1$}~[tnpos=l]{$CT(1,1)$}
\Tcircle{\scriptsize $n_1,T_0$}~[tnpos=r]{$CT(1,0)$}}
\caption{Construction of the first generation.}
\label{figure4}
\end{center}
\end{figure}

$\bullet$ If $CT(1,1) = \min\{ CT(1,0), CT(1,1) \}$, we select $%
(n_1,T_1) $ and  create its two children: $(n_2, T_2)$ and $(n_2,T_1)$.
Since $C_{rep}(T_1) < C_{rep}(T_0)$, we do not create the child of $(n_1, T_0)$, i.e., $(n_2, T_0)$ (see
Figure~\ref{figure5}).

\begin{figure}[h!]
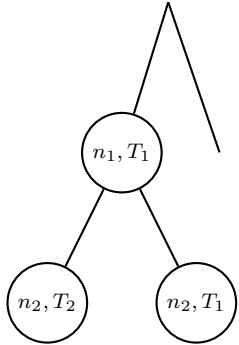

\begin{center}
\pstree{\Tr{}}  {\ \pstree{\Tcircle{\scriptsize
$n_1,T_1$}}{\Tcircle{\scriptsize $n_2,T_2$} \Tcircle{\scriptsize $n_2,T_1$}}
\Tr{}  }
\caption{Second generation when $CT(1,1) = \min\{CT(1,0), CT(1,1) \}$. }
\label{figure5}
\end{center}
\end{figure}

\n
Denote by $S(2)$ the set of all indices of deadlines at the second generation. Figure~\ref{figure5} gives $S(2) = \{ 1,2 \}$. Note that this set contains distinct deadlines. In the following, we will denote by $S(i)$ the set of indices of deadlines at the $i$th generation. 

$\bullet$ If $CT(1,0) = \min\{ CT(1,0), CT(1,1) \}$, we select the node $(n_1, T_0)$ to create its two children, $(n_2, T_2)$ and $(n_2, T_0)$. Since $%
C_{rep}(T_1) < C_{rep}(T_0)$, we generate the descendant of $(n_1,T_1)$, $(n_2, T_1)$ (see Figure~\ref{figure6}).
\begin{figure}[h!]
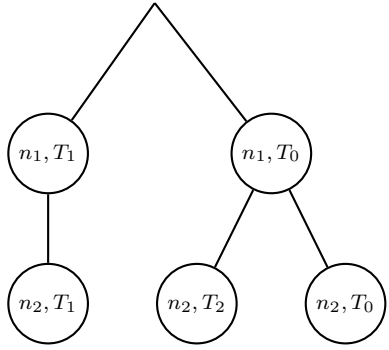

\begin{center}
\pstree{\Tr{}}
   { \pstree{\Tcircle{\scriptsize $n_1,T_1$}} {\Tcircle{\scriptsize $n_2,T_1$}}
   
     \pstree{\Tcircle{\scriptsize $n_1,T_0$}}{\Tcircle{\scriptsize $n_2,T_2$} \Tcircle{\scriptsize $n_2,T_0$}}
   } 
\end{center}
\begin{center}
\caption{ Second generation when $CT(1,0) = \min\{ CT(1,0), CT(1,1) \}$. }
\label{figure6}
\end{center} 
\end{figure}
In this case $S(2) = \{ 0,1,2 \}$. 

\subsubsection{Third generation}

In order to construct the third generation, we start from Figure~\ref{figure3}. As before, we will compare the total cost of the remaining repair programs, i.e., the branches $\{br(i),i=3,...,8\}$, in order to remove the most expensive ones. Let $CT(2,0)$, $CT(2,1)$ and $CT(2,2)$ be respectively the total cost of the partial programs $\{(n_{1}+ n_2,T_{0})\}$; $\{(n_{1}+ n_2,T_{1})\}$ and $\{(n_{1},T_{0}),(n_{2},T_{2})\}$. The nodes $(n_2,T_0)$, $(n_2,T_1)$ and $(n_2,T_2)$ are respectively labeled with $CT(2,0)$, $CT(2,1)$ and $CT(2,2)$ (see Figure~\ref{figure7}). These costs are given by
\begin{equation*}
\begin{array}{lll}
CT(2,0) = C(t,\{(n_1,T_0),(n_2,T_0)\}) = CT(1,0) + n_2C_{rep}(T_0),\vspace{%
0.1 cm} &  &  \\ 
CT(2,1) = C(t,\{(n_1,T_1),(n_2,T_1)\}) = CT(1,1) + n_2C_{rep}(T_1),\vspace{%
0.1 cm} &  &  \\ 
CT(2,2) = C(t,\{(n_1,T_0),(n_2,T_2)\}) = CT(1,0) + n_2C_{rep}(T_2) +
C_{out}(T_2). &  &  \\ 
&  & 
\end{array}%
\end{equation*}

\begin{figure}[h!]
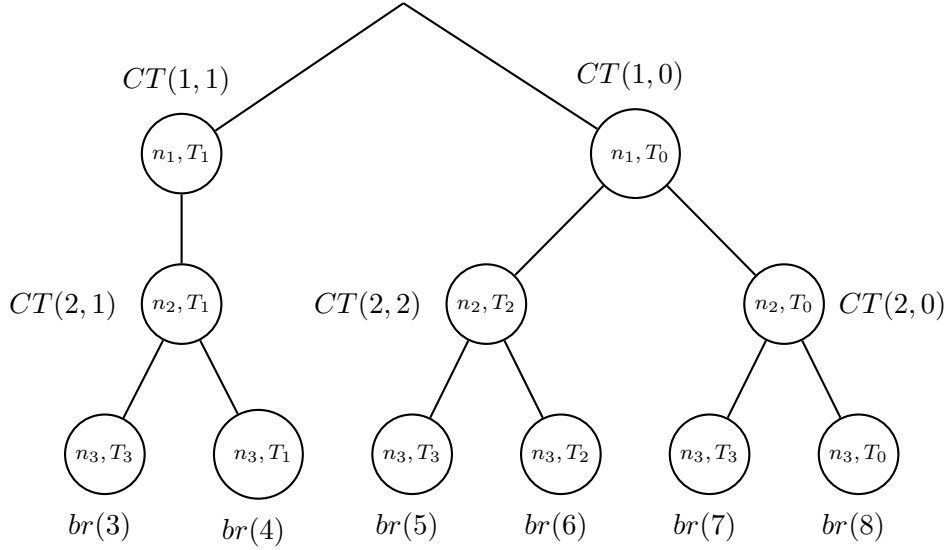

\begin{center}
\pstree{\Tr{}}
{ \pstree{\Tcircle{\scriptsize $n_1,T_1$}~[tnpos=a]{$CT(1,1)$}}
    { \pstree{\Tcircle{\scriptsize $n_2,T_1$}~[tnpos=l]{$CT(2,1)$}}{\Tcircle{\scriptsize $n_3,T_3$}~{$br(3)$} \Tcircle{    \scriptsize $n_3,T_1$}~{$br(4)$}}}

   \pstree{\Tcircle{ \scriptsize $n_1,T_0$}~[tnpos=a]{$CT(1,0)$}}{\pstree{\Tcircle{\scriptsize $n_2,T_2$}~[tnpos=l]{$CT(2,2)$}} {\Tcircle{\scriptsize $n_3,T_3$}~{$br(5)$} \Tcircle{\scriptsize $n_3,T_2$}~{$br(6)$}} \pstree{\Tcircle{\scriptsize $n_2,T_0$}~[tnpos=r]{$CT(2,0)$}}{\Tcircle{\scriptsize $n_3,T_3$}~{$br(7)$} \Tcircle{\scriptsize $n_3,T_0$}~{$br(8)$}}}
}
\end{center} 
\begin{center} 
\caption{Search tree when $CT(1,1) > CT(1,0)$ with labels $CT(2,0)$, $CT(2,1)$, $CT(2,2)$.}
\label{figure7}
\end{center}
\end{figure}

\n
The total costs of $\{br(i),i=3,...,8\}$ are respectively
\begin{equation*}
\begin{array}{lllll}
C(t,br(3)) = CT(2,1) + n_3C_{rep}(T_3) + C_{out}(T_3),\vspace{0.1 cm} &  & 
&  &  \\ 
C(t,br(5)) = CT(2,2) + n_3C_{rep}(T_3) + C_{out}(T_3),\vspace{0.1 cm} &  & 
&  &  \\ 
C(t,br(7)) = CT(2,0) + n_3C_{rep}(T_3) + C_{out}(T_3),\vspace{0.1 cm} &  & 
&  &  \\ 
C(t,br(4)) = CT(2,1) + n_3C_{rep}(T_1),\vspace{0.1 cm} &  &  &  &  \\ 
C(t,br(6)) = CT(2,2) + n_3C_{rep}(T_2),\vspace{0.1 cm} &  &  &  &  \\ 
C(t,br(8)) = CT(2,0) + n_3C_{rep}(T_0).\vspace{0.1 cm} &  &  &  &  \\ 
&  &  &  & 
\end{array}%
\end{equation*}

\n
We will compare the branches $br(3)$, $br(5)$ and $br(7)$, as well as the branches $br(4)$, $br(6)$ and $br(8)$. Once again we put conditions on $CT(2,0)$, $CT(2,1)$ and $CT(2,2)$, in order to be able to conclude on the costly programs. Below, we deal with only one case:
\begin{equation*}
CT(2,1) = \min( CT(2,0), CT(2,1), CT(2,2)).
\end{equation*}

\n 
When comparing the total costs of $br(3)$, $br(5)$ and $br(7)$, we obtain
\begin{equation*}
C(t,br(3)) < C(t,br(5)) \qquad \mbox{and} \qquad C(t,br(3)) < C(t,br(7)).
\end{equation*}
\n 
Thus, $br(5)$ and $br(7)$ are deleted, which is the same as removing the nodes $(n_{3},T_{3})$, which are the offspring of $(n_{2},T_{2})$ and $(n_{2},T_{0})$. Only one node $(n_{3},T_{3})$ is left at the third generation.
Since $C_{rep}(T_2) < C_{rep}(T_1) < C_{rep}(T_0)$, we have
\begin{equation*}
C(t,br(4))<C(t,br(8)).
\end{equation*}

\n 
Thus, we delete $br(8)$ and hence the node $(n_2,T_0)$, since we have previously removed the node $(n_3,T_3)$ (the descendant of $(n_2,T_0)$).  
Since $C_{rep}(T_{1})>C_{rep}(T_{2})$, we cannot make a conclusion about  $br(6)$.
Indeed,
\begin{equation*}
C(t,br(4)) - C(t,br(6)) = \underbrace{CT(2,1) - CT(2,2)}_{ < 0} + n_3%
\underbrace{(C_{rep}(T_1) - C_{rep}(T_2))}_{>0},
\end{equation*}

\n 
so we keep $br(6)$. Under the condition $CT(2,1)=\min (CT(2,0),CT(2,1),CT(2,2))$, the most expensive branches $br(5)$, $br(7)$ and $br(8)$, have been deleted. Thus, there remain only three PM schedules instead of eight, as depicted in Figure~\ref{figure8}.  

\begin{figure}
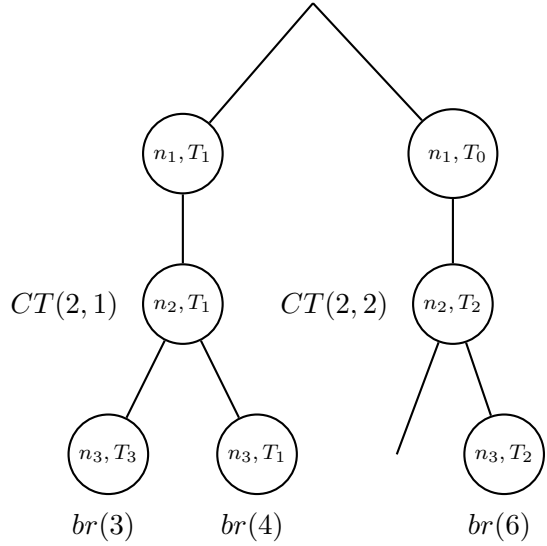

\begin{center}
\pstree{\Tr{}}
{ \pstree{\Tcircle{\scriptsize $n_1,T_1$}}
    { \pstree{\Tcircle{\scriptsize $n_2,T_1$}~[tnpos=l]{$CT(2,1)$}}{\Tcircle{\scriptsize $n_3,T_3$}~{$br(3)$} \Tcircle{\scriptsize $n_3,T_1$}~{$br(4)$}}}

   \pstree{\Tcircle{ \scriptsize $n_1,T_0$}}{\pstree{\Tcircle{\scriptsize $n_2,T_2$}~[tnpos=l]{$CT(2,2)$}} {\Tr{} \Tcircle{\scriptsize $n_3,T_2$}~{$br(6)$}} }
}
\caption{Final tree when $CT(1,1) > CT(1,0)$ and $CT(2,1) = \min\{ CT(2,0), CT(2,1), CT(2,2) \}$. }
\label{figure8}
\end{center}
\end{figure}

The last stage consists in determining the optimal solution for the fixed inspection time $t=T^*$. In order to do this, it suffices to evaluate the nodes of the last generation by adding to $CT(2,1)$, firstly $%
n_{3}C_{rep}(T_{3})+C_{out}(T_{3})$ (yielding the total cost of $br(3)$), and secondly $n_{3}C_{rep}(T_{1})$ (giving the total cost of $br(4)$). To get the total cost of $br(6)$, we add $n_{3}C_{rep}(T_{2})$ to $CT(2,2)$. The minimum value returns the optimum PM plan for the fixed inspection time $t=T^{\ast }$.

\begin{remark}
Figure~\ref{figure8} gives $S(3) = \{ 1,2,3 \}$.
\end{remark}

\begin{remark}
The third generation has been built with
\begin{enumerate}
\item The two offspring of the node $(n_{2},T_{1})$ (the node that gives the
minimum of $CT(2,i)$, $i=0,1,2$): 
$ (n_{3},T_{3})$ and $(n_{3},T_{1}). $

\item The offspring of the node $(n_{2},T_{2})$ (which satisfies $C_{rep}(T_{2})<C_{rep}(T_{1})$):  
$ (n_{3},T_{2})$. 
\end{enumerate}
\end{remark}

\begin{remark}
At each generation $i$, we build at most $(i+1)$ nodes $(n_{i},T_{j})$
for $j=0,\ldots ,i$. For example, Figure~\ref{figure6} shows that the second generation is composed of
nodes $(n_{2},T_{0})$, $(n_{2},T_{1})$ and $(n_{2},T_{2})$, and Figure~\ref{figure8}
shows that the third generation has been built with nodes $(n_{3},T_{1})$, $(n_{3},T_{2})$ and $(n_{3},T_{3})$.
\end{remark}

Let us summarize the building strategy related to the example introduced in Section~\ref{sub sec algo}. We have initially fixed the inspection time as $t=T^*$, and we wish to build a tree with three generations, since $N_t = 3$. The first generation $\{ (n_1,T_1), (n_1, T_0) \}$ is constructed and evaluated with $CT(1,1)$ and $CT(1,0)$. The node with the lowest cost generates its two children. The other node generates its descendant, provided that its repair cost calculated at its deadline is less than the repair cost calculated at the deadline of the least expensive node. The result of this first step is the second generation. At the end of the second stage, the third generation is built in the same manner as before, and gives the final tree containing three generation. The aim of the third and final stage is to determine the optimal repair plan for the given inspection time; for this, it suffices to evaluate the nodes at the last generation. The branch that returns the least expensive cost is the best repair schedule. Recall that this PM plan is not necessarily the global optimum, it could be a local minimum. In order to obtain the global optimum, we have to build trees related to inspection times $T_2 -1$ and $T_3 -1$ (the case $t=T_1 -1$ is trivial because $N_{t} = 0$, i.e., there are no repairs, hence no tree). Each of such tree returns the best repair program. To get the (global) PM schedule, we have to compare the total costs of these best repair plans. As a result, with our algorithm,  pipeline managers have the opportunity to choose one of these PM schedules. If they choose to inspect at year $T^*$, the algorithm will output the least expensive repair plan for this choice. The general algorithm is presented in the next section.

\subsection{Algorithm}
\label{section:algorithm}

The following algorithm generalizes the previous example. It provides a search tree with $j-1$ generations for a fixed inspection time $t=T_{j} - 1$. The construction of the remaining trees follows the same algorithm. At the end, we obtain $N+1$ PM schedules, i.e., an inspection time together with the best related repair plan; the global minimum is obtained by comparing the total costs of these $N+1$ repair plans. \vspace{0.5cm}
\begin{algorithm2e}
\KwResult{Optimal PM schedule for an inspection at year $t=T_j - 1$.}
\textbf{Initialization}: Construction of the first generation $\{(n_{1},T_{1}),(n_{1},T_{0})\}$\;
\hspace*{2.6cm} Evaluation of partial repair programs $CT(1,1)$ and $CT(1,0)$\;
\textbf{Step 1: second generation}\;
 \eIf{$\boldsymbol{CT(1,1) < CT(1,0)}$}{
Both children of node $(n_1,T_1)$, i.e., $\{(n_{2},T_{2}),(n_{2},T_{1})\}$ are constructed\;
 \If{$C_{rep}(T_{0})>C_{rep}(T_{1})$}{it is not necessary to build the
offspring of $(n_{1},T_{0})$}
The second generation is built\;
}
{
Both descendants of $(n_{1},T_{0})$, i.e., $\{(n_{2},T_{2}),(n_{2},T_{0})\}$ are built\;
\If{$C_{rep}(T_{0})>C_{rep}(T_{1})$}{ the descendant of $(n_{1},T_{1})$, i.e., $(n_{2},T_{1})$, is built;}
The second generation is built;
}
\textbf{Step i: $\boldsymbol{(i+1)}$th generation, $\boldsymbol{1< i < j-1}$}\;
\hspace*{1.3cm} The $i$th generation was built at the previous step: $\{(n_{i},T_{l}),l\in S(i)\}$\;
\If{$\boldsymbol{CT(i,\tilde{l})=\min \{CT(l),l\in S(i)\}}$}{
Both children of node $(n_1,T_{\tilde{l}})$ are constructed, i.e., $\{(n_{i+1},T_{i+1}),(n_{i+1},T_{\tilde{l}})\}$\;
\While{$\boldsymbol{C_{rep}(T_k) < C_{rep}(T_{\tilde{l}})}$, $\boldsymbol{k \in S(i) - \{\tilde{l}}$\}}{Construction of the descendant of $(n_k,T_k)$, i.e.,$(n_{i+1},T_{k})$}
The $(i+1)$th generation is built\;
}
\textbf{Step $\boldsymbol{j-1}$: Evaluation of full repair plans}\;
\hspace*{1.9cm} The best PM schedule for an inspection at $t=T_{j}-1$ is given.
\vspace{.4cm}

\caption{Construction of the tree with $j-1$ generations.}
\label{algorithm}
\end{algorithm2e}

This algorithm builds and estimates the costs of at most $(N+1)(N+2)/2 - 1$ repair schedules. Indeed, if $t=T_{1}-1$, the optimal program is that where there is no repair. In this case, the total cost coincides with the inspection cost evaluated at time $t$. If $t=T_{2}-1$, the algorithm builds and estimates the costs of at most two branches, and returns the best repair schedule, $b^*(T_{2}-1)$. If $t=T^{\ast }$, the algorithm builds and estimates the costs at most $N+1$ branches, and returns the best repair schedule, $b^*(T^*)$. Thus, the optimal maintenance program $b^{\ast }(t^{\ast })$ is  the best among $\left\lbrace b^*(T_2-1),\ldots ,b^*(T^*). \right\rbrace$

\section{Computational results}
\label{section4:simulations}

In this section, we present a large number of examples to confirm the efficiency of the proposed algorithm. The characteristics of the 11 examples are shown in Table~\ref{tab1}. The column labeled with $N$ corresponds to the number of distinct deadlines, the third column designates the deadlines, and the last column the numbers of defects associated with each deadline. In other words this table gives $11$ \textit{primary repair schedules} defined in Section~\ref{section1:problem description} (Definition~\ref{primary PM}). 
For example, the the fifth experience is composed of $7$ different deadlines. The first four are $T_1=2$, $T_2=5$, $T_3=8$ and $T_4=15$. At year $T_1=2$ there is $1$ repairs; at year $T_2=5$ there is $1$ repair; at year $T_3=8$, $1$ and at year $T_4=15$, $1$. Hence the primary is given by
\begin{align*}
\begin{split}
\mathscr{P} &= \left\lbrace (n_1=2,T_1=2), (n_2=1,T_2=5), (n_3=1,T_3=8), (n_4=1,T_4=15) \right. \vspace{0.3cm}\\
 \quad & \left. (n_5=6,T_5=24), (n_6=5,T_6=26), (n_7=4,T_7=28) \right\rbrace.
\end{split}
\end{align*}
\n For the above primary repair schedule, the PM problem aims at finding the next (after the primary inspection) optimal inspection $t^*$ within 
\begin{equation*}
\left\lbrace T_j - 1, \, j=1, \ldots , N, T^* \right\rbrace = \left\lbrace 1, 4, 7, 15, 24, 25, 27, 30 \right\rbrace,
\end{equation*}
and the associated repair schedule minimizing the operational cost within $\Delta t^*$; here $30$ designates the horizon time $T^*$. For this experience, our algorithm built at most $35 = (7+1)(7+2)/2-1$ repairs schedule and return $7$ maintenance programme; one for each epoch: $4, 7, 15, 24, 25, 27$, and $30$.
\vspace{0.3cm}

We want to look at the computational effort required to solve the problem \eqref{Opt Problem} when the number of deadlines $N$ increases. To do this, we will compare our algorithm, which we call \textit{tree algorithm}, with two other algorithms. The first, called the \textit{simplified algorithm}, solves \eqref{Opt Problem} and looks through complete repair plans when  candidates for the next inspection are in $\{ T_1-1, \ldots , T_N-1, T^*\}$. According to Section~\ref{section:algorithm}, this method investigates $(1-2^{N+1})/(1-2) - 1$ schedules. The second one, called the \textit{comprehensive algorithm}, solves the same problem but investigates $2^{N_t}$ schedules for each $t \in \{1,2,\ldots , T^*  \}$. This method provides an inefficient but complete set of PM schedules (inspection times and repair plans) which will serve as a benchmark for our algorithm and the simplified one. Both algorithms (simplified and comprehensive) output the best repair schedule for each inspection time $t \in \{ T_1-1, \ldots , T_N -1, T^* \}$ (and especially, the optimal PM plan $(t^*, \mathscr{P}^{*}_t)$). 
\vspace{0.3cm}

The experimental design consists of entering manually for the three algorithms, the number of different deadlines $N$, the values of deadlines, and the number of defects observed for each deadline. The parameters $T^{\ast }$,  discount rate, and inflation rate, are also entered manually and  respectively set as $T^* = 30$, $r_d = 8\%$, $r_i = 1\%$. The full costs are expressed in $k$\euro{} and are also manually entered. The initial inspection cost is $C_{insp}^{0}=500$, the initial repair cost is $C_{rep}^{0}=60$, and the initial out-of-service cost is $C_{out}^{0}=300$. We compare the performance of our algorithm with the comprehensive and  simplified algorithms. The computational tests were built with scilab 5.4.1 on a SONY computer with biprocessor, 2.30 GHz and 1 GB of RAM. The three algorithms were developed in a decision support framework. The tools provide a set of suitable solutions. Each solution is defined by a repair schedule, an inspection time belonging to $\{T_{j}-1,j=1,\ldots ,N\}\cup \{T^{\ast }\}$, and a total cost. Table~\ref{tab2} illustrates the results. It compares the running times (indicated by the CPU) and the optimal cost of the all $11$ experiments for each algorithm. 

When we compare the outputs of our algorithm with both the comprehensive and simplified algorithms, we see that all the algorithms return the same optimal cost for each experience, whereas the running time is much less than for the others.
Our algorithm was developed under a support decision framework. We have seen that its results can help a pipeline manager, in the sense that the algorithm provides a set of solutions that includes the optimal PM schedule. Table~\ref{tab3} below shows for $N=7$, all PM schedules, i.e., all repair schedules associated with $t \in \{T_{j}-1,j=1,\ldots,7 \}\cup \{ T^*\}$. 
The first column corresponds to the distinct inspection times, the second to the associated repair plans, and the last represents the total cost of these repair schedules within the inspection interval. 
We see that the optimal PM schedule suggests an inspection at year 23, for a total cost of $347.05704$ $k$\euro{}. The tool gives then to the pipeline manager the following repair program within the inspection interval $]0,23]$
\[\{(4,T_0),(0,T_1), (0,T_2),(0,T_3),(0,T_4)\}.\]
This repair schedule suggests to the practitioner to repair $4$ defect during the primary inspection (at $T_0 = 0$)) and no repair at years $T_1 = 2$, $T_2=5$, $T_3=8$ and $T_4 = 15$.
If the operator is not satisfied with this PM schedule, he has the opportunity to select another schedule. For example, he may choose to inspect later at year $27$. In this case, the tool proposes a repair schedule with a total cost equal to $514.11211$ $k$\euro{}, and the repair plan is given by
\[ \{(4,T_0),(0,T_1), (0,T_2), (0,T_3), (0,T_4), (11,T_5), (0,T_6) \},\]
which means that $4$ defect must be repair during the primary inspection and $11$ repairs should be made at year $T_5 =24$.

\begin{center}
\begin{tabular}{|c|c|l|l|}
\hline
Example N\textsuperscript{o} & $N$ & Deadlines $T_j$, $j=1, \ldots N$ & Number of defects \\ \hline
1  & 3  & 1,8,16 & 1,2,1 \\ \hline
2  & 4  & 4,6,12,22 & 2,3,1,1 \\ \hline
3  & 5  & 2,3,8,12,24 & 1,1,3,2,2 \\ \hline
4  & 6  & 5,6,8,15,19,25 & 1,1,3,3,2,1 \\ \hline
5  & 7  & 2,5,8,15,24,26,28 & 1,1,1,1,6,5,4 \\ \hline
6  & 8  & 4,7,8,11,13,21,25,27 & 1,1,2,1,1,3,3,1 \\ \hline
7  & 9  & 5,6,8,11,14,20,21,25,26 & 1,2,3,2,1,1,3,2,1 \\ \hline
8  & 10 & 3,5,6,7,13,18,20,22,25,26 & 1,2,1,2,3,3,1,1,1,1 \\ \hline
9  & 11 & 2,4,5,6,10,12,16,17,20,22,25 & 1,2,3,1,4,1,1,1,2,2,3 \\ \hline
10 & 12 & 2,4,5,6,7,9,10,11,18,20,24,26 & 1,1,2,3,2,2,1,1,1,3,2,1 \\ \hline
11 & 13 & 2,5,6,7,10,12,13,17,18,20,21,24,26 & 1,1,3,2,2,1,1,3,4,1,4,5,2 \\ 
\hline
\end{tabular}
\captionof{table}{Data for simulated examples.}
\label{tab1}
\end{center}

\begin{center}
\begin{tabular}{|l|cc|ccc|ccc|}
\hline
\multicolumn{1}{|c|}{$N$} & \multicolumn{2}{l|}{Comprehensive algorithm}  & & \multicolumn{2}{l|}{Simplified Algorithm} 
& & \multicolumn{2}{l|}{Tree Algorithm} 
\\ 
\cline{2-3} \cline{4-6} \cline{7-9}
& CPU & Optimal cost &  & CPU & Optimal cost & & CPU & Optimal Cost (\euro{})\\ 
\hline
3  & 0.60     & 306972.81  &  & 0.078  & 306972.81  & & 0.0156  & 306972.81 \\ 
4  & 1.61     & 408943.08  &  & 0.22   & 408943.08  & & 0.0156  & 408943.08 \\ 
5  & 3.26     & 432790.34  &  & 0.44   & 432790.34  & & 0.0312  & 432790.34 \\ 
6  & 9.39     & 382437.51  &  & 1.37   & 382437.51  & & 0.0312  & 382437.51 \\ 
7  & 24.60    & 347057.04  &  & 3.88   & 347057.04  & & 0.0468  & 347057.04 \\ 
8  & 54.40    & 394468.88  &  & 9.42   & 394468.88  & & 0.0624  & 394468.88 \\ 
9  & 123.27   & 382437.51  &  & 24.08  & 382437.51  & & 0.0624  & 382437.51 \\ 
10 & 272.24   & 437285.67  &  & 57.24  & 437285.67  & & 0.0780  & 437285.67 \\ 
11 & 545.78   & 467592.59  &  & 124.16 & 467592.59  & & 0.0780  & 467592.59 \\ 
12 & 1138.46  & 467592.59  &  & 278.82 & 467592.59  & & 0.1092  & 467592.59 \\ 
13 & 2656.50  & 442437.51  &  & 725.95 & 442437.51  & & 0.1248  & 442437.51 \\ 
\hline
\end{tabular}
\captionof{table}{Computational results.}
\label{tab2}
\end{center}

%

\begin{center}
\begin{tabular}{|c|l|c|}
\hline
Inspection time  & Repair within $\Delta t$ plans & Total cost \\ 
\hline
30  &  $\{(4,T_0),(0,T_1), (0,T_2), (0,T_3), (0,T_4), (15,T_5), (0,T_6), (0,T_7)  \}$ & 547256.39  \\
\hline
27   & $\{(4,T_0),(0,T_1), (0,T_2), (0,T_3), (0,T_4), (11,T_5), (0,T_6)  \}$          & 514112.11 \\ 
\hline
25   & $\{(4,T_0),(0,T_1), (0,T_2), (0,T_3), (0,T_4), (6,T_5) \}$                    & 465784.98\\ 
\hline
23   & $\{(4,T_0),(0,T_1), (0,T_2), (0,T_3), (0,T_4)  \}$  							 & 347057.04 \\ 
\hline
14   & $\{(53,T_0),(0,T_1), (0,T_2), (0,T_3) \}$                                      & 375675.59 \\ 
\hline
7    & $\{(2,T_0),(0,T_1), (0,T_2)  \}$                                              & 432790.34 \\ 
\hline
4   & $\{(1,T_0),(0,T_1) \}$                                                         & 442437.51\\ 
\hline
1  & $\{(0,T_0)\}$                                                                   & 467.59259\\ 
\hline
\end{tabular}
\captionof{table}{Data for simulated examples when $N=7$.}
\label{tab3}
\end{center}

\section{Conclusion}
\label{section:conclusion}

This paper focused on the PM problem in gas pipelines from the economics point of view. We have modeled this problem using binary integer nonlinear programming, which drastically shrinks the possible set of good repair schedules, by exploiting the problem's properties. In order to solve the problem, we have proposed an algorithm, based on dynamic programming, which finds  the exact solution extremely quickly, along with a set of alternative solutions. As a result, managers of gas pipeline systems can consider various possible schedules, and choose alternative PM programs, and our algorithm assists them in making the most suitable decision in a short period of time.

\bibliographystyle{plain}
\bibliography{bibliography}

\end{document}